\documentclass[letterpaper,10pt, conference]{_Aux/IEEEconf}
\IEEEoverridecommandlockouts
\usepackage{times}
\usepackage{setspace}
\usepackage{url}
\spacing{1}
\usepackage[utf8]{inputenc}
\usepackage[T1]{fontenc}
\usepackage{graphicx}		
\usepackage{wrapfig}
\usepackage[format=plain,font=footnotesize,labelfont=bf,labelsep=period]{caption}
\usepackage{sidecap} 
\usepackage[export]{adjustbox}
\usepackage{subcaption}
\usepackage[font=small]{caption}
\usepackage{float}

\usepackage{amsmath} 
\usepackage{amssymb}  
\usepackage{amsthm}
\usepackage{mathtools}
\usepackage[normalem]{ulem}
\usepackage{paralist}	
\usepackage[space]{grffile} 
\usepackage{color}
\let\labelindent\relax
\usepackage{enumitem}
\usepackage{bm}
\usepackage{cancel}
\usepackage{hhline}
\usepackage[c2 , nocomma]{optidef}

\newtheorem{theorem}{Theorem}
\newtheorem{corollary}{Corollary}
\newtheorem{lemma}{Lemma}
\theoremstyle{definition}
\newtheorem{definition}{Definition}
\theoremstyle{remark}
\newtheorem{remark}{Remark}
\theoremstyle{definition}
\newtheorem{assumption}{Assumption}
\theoremstyle{definition}

\newtheorem{example}{Example}

\newcommand{\C}{\mathcal{C}}

\newcommand{\U}{\mathbb{U}}

\definecolor{blue}{RGB}{38,38,134}
\definecolor{darkblue}{RGB}{0,0,102}
\definecolor{lightblue}{RGB}{77,77,148}

\definecolor{gold}{RGB}{234, 170, 0}
\definecolor{metallic_gold}{RGB}{139, 111, 78}

\newcommand{\norm}[1]{\left\Vert #1 \right\Vert}

\newcommand{\derp}[2]{\frac{\partial #1 }{\partial #2 }}

\newcommand{\intersect}{\cap}

\DeclareMathOperator{\diag}{diag}

\usepackage{xfrac}

\usepackage{dblfloatfix}

\usepackage{pdfpages}
\usepackage{graphicx}

\usepackage{cite}

\begin{document}

\title{ \bf
Robust Control Barrier Functions with Uncertainty Estimation 
}

\author{Ersin Da\c{s}$^{1}$, Skylar X. Wei$^{1}$, Joel W. Burdick$^{1}$
\thanks{*This work was supported by DARPA under the LINC program.}
\thanks{$^{1}$E. Da\c{s}, S. X. Wei, and J. W. Burdick are with  the
Department of Mechanical and Civil Engineering, California Institute of
Technology, Pasadena, CA 91125, USA. ${\tt\small \{ersindas, swei, jburdick \}@caltech.edu}$ } }

\maketitle

\pagestyle{plain}

\begin{abstract}
This paper proposes a safety controller for control-affine nonlinear systems with unmodelled dynamics and disturbances to improve closed-loop robustness.   Uncertainty estimation-based control barrier functions (CBFs) are utilized to ensure robust safety in the presence of model uncertainties, which may depend on control input and states.
We present a new uncertainty/disturbance estimator with theoretical upper bounds on estimation error and estimated outputs, which are used to ensure robust safety by formulating a convex optimization problem using a high-order CBF. The possibly unsafe nominal feedback controller is augmented with the proposed estimator in two frameworks (1) an uncertainty compensator and (2) a robustifying reformulation of CBF constraint with respect to the estimator outputs. The former scheme ensures safety with performance improvement by adaptively rejecting the matched uncertainty. The second method uses uncertainty estimation to robustify higher-order CBFs for safety-critical control. The proposed methods are demonstrated in simulations of an uncertain adaptive cruise control problem and a multirotor obstacle avoidance situation.
\end{abstract}

\section{Introduction} \label{sec:intro}
Safety-critical controller synthesis is a challenging problem in the deployment of autonomous and cyber-physical systems \cite{guiochet2017,9536448}. Control barrier functions (CBFs) ensure safety by certifying forward invariance of a safe set \cite{ames2017control}. Thanks to the linearity of the CBF constraint for nonlinear control-affine systems, safety requirements can be enforced with a CBF-QP to optimize for pointwise control input \cite{ames2017control}.

Because known dynamics models are needed to synthesize the CBF constraints, the standard formulation may be sensitive to inevitable model uncertainties and disturbances. Numerous robust CBF approaches have been proposed to address this issue \cite{jankovic2018robust, kolathaya2018CLF, nguyen2021robust}. However, these methods use a bound to represent the unmodelled dynamics, which may be difficult to estimate in practice, and usually result in undesired conservativeness and a reduction in performance.   

Recently, connections between disturbance observer-based techniques, a well-studied robust control tool \cite{dacs2022combined}, and CBFs have been made to actively estimate and compensate for external disturbance to guarantee robust safety \cite{zhao2020adaptive, dacs2022robust, wang2022disturbance, cheng2022safe, alan2022disturbance}. Despite promising results in joining CBFs with disturbance observers, prior studies have been limited to state-dependent external disturbances. However, in many applications, uncertainties may depend on both states and control. Furthermore, these studies have not used active disturbance/uncertainty compensation capabilities to purposefully design observers that improve control performance and guarantee safety.   

Motivated by these limitations, we propose a more general uncertainty/disturbance estimation-based robust CBF framework for control-affine nonlinear systems with state and control input dependent uncertainties. We introduce a new uncertainty/disturbance estimator to observe unmodeled dynamics in multiple-input-multiple-output nonlinear systems.  Upper bounds for the estimation error and estimated output are developed under boundedness and local Lipschitz continuity assumptions on the uncertainty. Then, two uncertainty estimator-based robust CBFs methods are proposed. Method 1: when a matched uncertainty/disturbance dynamics condition is met, estimated values are used to compensate for these dynamics and adaptively modify the CBF condition against unmodelled system dynamics (see Fig.~\ref{fig:scheme}, top). Method 2 integrates the proposed estimator with higher-order CBFs to guarantee robustness for high relative degree safety constraints and to obtain an augmented safe set (see Fig.~\ref{fig:scheme}, bottom). To showcase our results, we simulate an uncertain adaptive cruise control example (Method 1) and a multirotor obstacle avoidance scenario (Method 2).
\begin{figure}[t]
	\centering
	\includegraphics[scale=0.65]{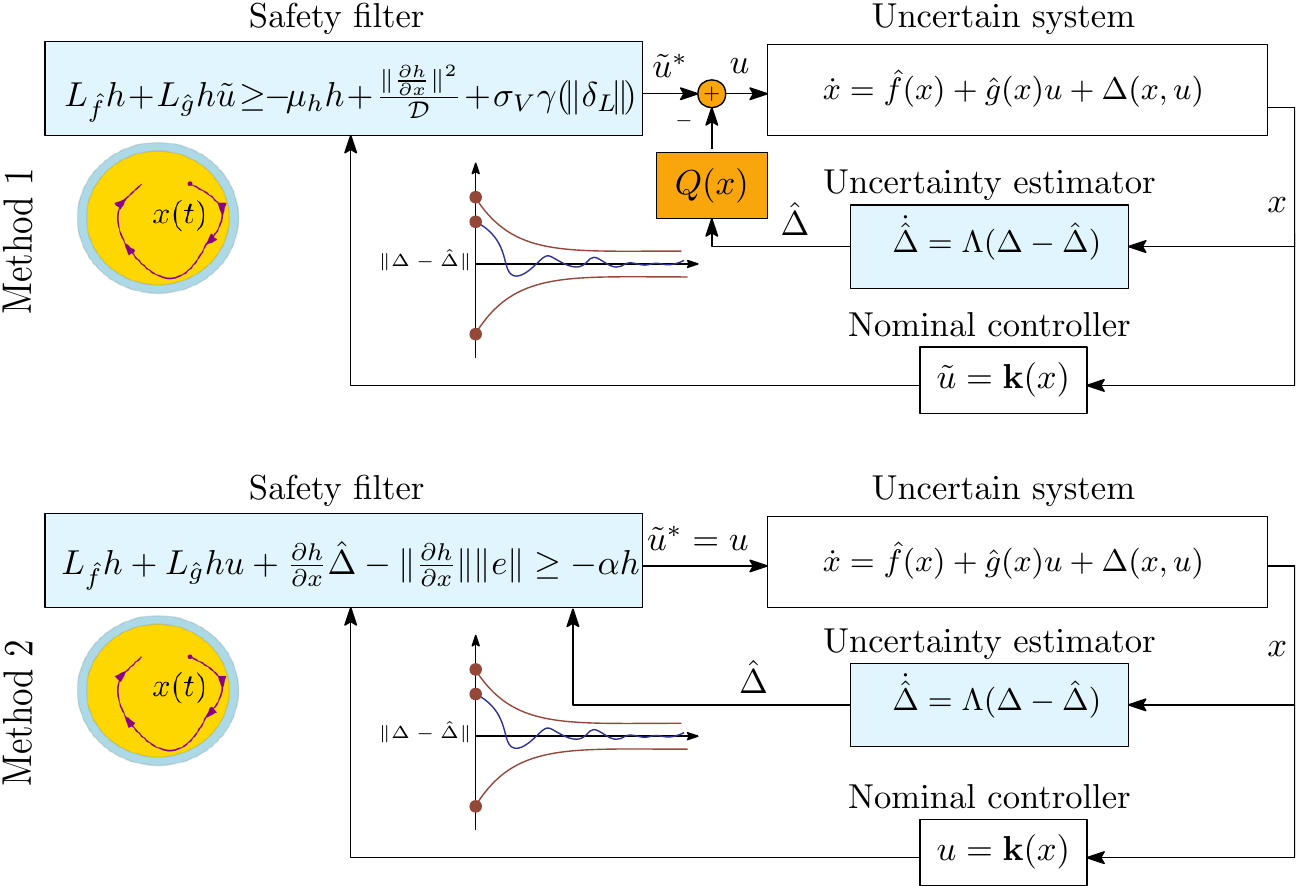}
	\caption{A block diagram of the proposed uncertainty estimator-based safe control frameworks. Augmenting a given, and potentially unsafe nominal controller with an error-bounded uncertainty estimator and safety filter, guarantees that the uncertain nonlinear system states remain in an inner safe set.}
	\label{fig:scheme}
	\vspace{-4mm}
\end{figure}

This paper is organized as follows. After preliminaries are introduced in Section~\ref{sec:pre}, Section~\ref{sec:pro} states our problem of safety-critical control under unmodeled dynamics. Section~\ref{sec:ue} proposes and analyzes an uncertainty estimator, and introduces estimator-based robust CBFs schemes. Simulations are presented in Section IV, and Section V concludes the paper.

\section{Preliminaries}
\label{sec:pre}
\textit{Notation:} ${\mathbb{R}, \mathbb{R}^+, \mathbb{R}^+_0}$ represent the set of real, positive real, and non-negative real numbers, respectively.  The Euclidean norm of a matrix is denoted by $\|\cdot\|$. 
For a given set ${\mathcal{C}  \subset \mathbb{R}^n}$, ${\partial \mathcal{C}}$ and Int$(\mathcal{C})$ denote its boundary and interior, respectively.

A continuous function ${\alpha : \mathbb{R}^+_0 \rightarrow \mathbb{R}^+_0}$ belongs to class-${\mathcal{K}}$ (${\alpha \in \mathcal{K}}$) if it is strictly monotonically increasing and ${\alpha(0) = 0}$. Further, a continuous function ${\alpha : \mathbb{R}^+_0 \rightarrow \mathbb{R}^+_0}$ belongs to class-${\mathcal{K}_\infty}$ (${\alpha \in \mathcal{K}_\infty}$) if it is strictly increasing, ${\alpha(0) = 0}$, and ${\alpha(r) \rightarrow \infty}$ as ${r \rightarrow \infty}$. A continuous function ${\alpha : \mathbb{R} \rightarrow \mathbb{R}}$ belongs to the set of extended class-$\mathcal{K}$ functions (${\alpha \in \mathcal{K}_{\infty, e}}$) if it is strictly monotonically increasing, ${\alpha(0) = 0}$, ${\alpha(r) \rightarrow \infty}$ as ${r \rightarrow \infty}$, and ${\alpha(r) \rightarrow -\infty}$ as ${r \rightarrow -\infty}$. Lastly, a continuous function ${\beta(\cdot, \cdot) : \mathbb{R}^+_0 \times \mathbb{R}^+_0 \to \mathbb{R}^+_0}$ belongs to class-${\mathcal{K L}}$ denoted by (${\beta \in \mathcal{KL}}$), if for every ${s \in \mathbb{R}^+_0}$, ${\beta(\cdot, s) : \mathbb{R}^+_0 \to \mathbb{R}^+_0}$ is a class-$\mathcal{K}$ function and for every ${r \in \mathbb{R}^+_0}$, ${\beta(r, \cdot) }$ is decreasing and ${\lim_{s \to \infty} \beta(s, r) = 0}$.

In this work, we consider a nonlinear control affine system:
\begin{equation}
\label{system}
    \dot{x}  = f(x) + g(x) u,
\end{equation}
where ${x \in X \subset \mathbb{R}^n}$ is the states, ${u \in U \subset \mathbb{R}^m}$ is the control input, and ${f: X \rightarrow \mathbb{R}^n}$, ${g: X \rightarrow \mathbb{R}^{n \times m} }$ are locally Lipschitz continuous functions. Given an initial state, ${x_0 \in X}$ and a locally Lipschitz continuous controller ${u = \mathbf{k}: X \to U}$, which yields a locally Lipschitz continuous closed-loop system, there exists a unique solution ${x(t),~\forall t \geq 0}$ satisfying the closed-loop dynamics and initial state. We assume ${f(0) = 0}$. Throughout this paper, we call~(\ref{system}) the \textit{actual (uncertain, real) model}.


\subsection{Control Barrier Functions}
\label{sec:cbfs}
We consider a set ${\mathcal{C} \subset X }$ defined as a 0-superlevel set of a continuously differentiable function ${h: X \rightarrow \mathbb{R}}$ as
\begin{align}
\label{CBF1}
    \mathcal{C} \triangleq \left\{ x \in X \subset \mathbb{R}^n : h(x) \geq 0 \right\}, \\
    \label{CBF12}
    \partial \mathcal{C} \triangleq \left\{ {x \in X \subset \mathbb{R}^n} : h(x) = 0 \right\}, \\
    \label{CBF13}
    \text{Int}(\mathcal{C}) \triangleq \left\{ {x \in X \subset \mathbb{R}^n} : h(x) > 0 \right\}.
\end{align}
This set is forward invariant if, for every initial condition ${x(0) \in \mathcal{C}}$, the solution of (\ref{system}) satisfies ${x(t) \in \mathcal{C}, ~\forall t \geq 0}$. Then, the closed-loop system (\ref{system}) is safe on the set $\mathcal{C}$ if $\mathcal{C}$ is forward invariant \cite{ames2017control}.

\begin{definition}[Control Barrier Function \cite{ames2017control}]
\label{def:cbf}
Let ${\mathcal{C} \subset X }$ be the 0-superlevel set of a continuously differentiable function ${h: X \rightarrow \mathbb{R}}$. 
Then, $h(x)$ is a \textit{Control Barrier Function} for system (\ref{system}) on $\mathcal{C}$ if ${\frac{\partial h}{\partial x} \neq 0}$ for all ${ x \in \partial \mathcal{C}}$ and there exists $\alpha \in \mathcal{K}_{\infty, e}$ such that $\forall x \in \mathcal{C}$:
\begin{align}
\label{cbf}
   \sup_{u \in U}\! \Big ( \!{\dot{h}(x, u)}  \!=\!
   \underbrace{ \dfrac{\partial h(x)}{\partial x}\! f(x) }_{L_f h(x)} \!+ \!\underbrace { \dfrac{\partial h(x)}{\partial x}g(x) }_{L_g h(x)} \!u \!\Big ) 
   \!\geq\! -\alpha (h(x)).
\end{align}
\end{definition}
\begin{theorem}\cite{ames2017control}
    \label{teo:cbfdef}
    If $h(x)$ is a valid CBF for \eqref{system} on $\mathcal{C}$, then any Lipschitz continuous controller ${u= \mathbf{k}: X \to U}$ satisfying 
    \begin{equation}
        \label{eq:cbf_def}
        \dot{h}\left(x, \mathbf{k}(x)\right) \geq - \alpha (h(x)),~~\forall x \in \mathcal{C}
    \end{equation}
    results in the safety of set $\mathcal{C}$ for \eqref{system}. 
\end{theorem}
Given a nominal (potentially unsafe) locally Lipschitz continuous controller $\mathbf{k_d}: X \to U$ and a CBF ${h(x)}$ for system~(\ref{system}), the safety can be ensured by solving the following CBF-Quadratic Program (CBF-QP):   
\begin{argmini}|s|
{u \in U}{\|u-\mathbf{k_d}(x)\|^2}
{\label{CBF-QP}}
{\mathbf{k}=}
\addConstraint{\dot{h}(x, u)  \geq - \alpha (h(x)) }
\end{argmini}

\subsection{Model Uncertainty}

Consider the \textit{nominal} control affine system:
\begin{equation}
\label{sysnom}
    \dot{x}  = \hat{f}(x) + \hat{g}(x) u, 
\end{equation}
where ${\hat{f}: X \rightarrow \mathbb{R}^n}$, ${\hat{g}: X \rightarrow \mathbb{R}^{n \times m} }$ are locally Lipschitz continuous are known for controller design. To define the discrepancies between the actual model and nominal model, we add and subtract \eqref{sysnom} to \eqref{system}:
\begin{equation}
\label{sysun}
    \dot{x}  = \hat{f}(x) + \hat{g}(x) u + \underbrace{ {f}(x)- \hat{f}(x) }_{\Delta f(x)} + \underbrace { ( {g}(x)- \hat{g}(x) ) }_{\Delta g(x)} u,
\end{equation}
where ${\Delta f: X \to \mathbb{R}^n,~\Delta g:  X \to \mathbb{R}^{n \times m} }$ are the unmodelled system dynamics, i.e., model uncertainties.
Now, we have the time derivative of CBF for the uncertain system as:
\begin{equation}
    \label{eq:hdot}
    \!\!{\dot{h}(x, u)}  \!=\!
   \dfrac{\partial h}{\partial x} ( \hat{f}(x) \!+\! \hat{g} (x) u )  \\ \!+\!  \dfrac{\partial h}{\partial x}  \Delta f(x)   
   \!+\!  \dfrac{\partial h}{\partial x} \Delta g(x)  u.
\end{equation}

\begin{assumption}
\label{as1}
The augmented uncertainty function 
\begin{equation}
    \label{delta}
    \Delta(x,u) := \Delta f(x) + \Delta g(x) u,
\end{equation}
is locally Lipschitz continuous for all ${t \geq 0}$ with Lipschitz constant $\delta_L$, and bounded by ${\| \Delta(x,u) \| \leq \delta_b}$ for all ${(x,u) \in X\times U}$. Further, we assume ${\delta_L, \delta_b \in \mathbb{R}^+}$ are known constants.
\end{assumption}

\begin{definition}[Input Relative Degree] The \textit{input relative degree (IRD)} of a sufficiently differentiable output function $h: X \rightarrow \mathbb{R}$ of system \eqref{system} is defined as an integer ${r \leq n}$ if ${\forall {x \in X}}$:
\begin{align}
\begin{split}
    \label{ird_ird}
    { L_{{g}}  L_{{f}}^{r-1} h(x) \neq 0} , \\
    L_{{g}} h(x) = L_{{g}}  L_{{f}} h(x) = \cdots = L_{{g}}  L_{{f}}^{r-2} h(x) = 0, 
\end{split}
\end{align}  
\end{definition}
\begin{definition}[Disturbance Relative Degree] The \textit{disturbance relative degree (DRD)} of a sufficiently differentiable output function $h: X \rightarrow \mathbb{R}$ of the uncertain system \eqref{sysun} with the augmented uncertainty function \eqref{delta} is defined as ${v \leq n}$ if ${\forall {x \in X}}$:
\begin{align}
\begin{split}
    \label{ird_drd}
    {\derp{L_{\hat{f}}^{r-1} h(x)}{x} \Delta(x, u) \neq 0}, \\
    \derp{h(x)}{x} \Delta(x, u) =  \cdots = \derp{L_{\hat{f}}^{r-2} h(x)}{x} \Delta(x, u) = 0, 
\end{split}
\end{align}  
\end{definition}
The uncertainty in \eqref{sysun} is {\em mismatched} if {DRD < IRD} and {\em matched} if {DRD = IRD}. Practically speaking, the matching condition implies the disturbance, and the control input appears on the same output channel of the system \cite{yang2012}. 
\begin{assumption}
    System \eqref{system} with a sufficiently differentiable output function $h:X\to \mathbb{R}$ satisfies {IRD = DRD}, i.e., {\em matched} uncertainty.
\end{assumption}

\subsection{Higher Order Control Barrier Functions}
\label{sec:CBF}
The definitions of CBFs in Section~\ref{sec:cbfs} require a relative degree one ${h(x)}$. However, for some safety requirements, we may need to differentiate function $h(x)$ with respect to the system (\ref{system}) until the control input appears. In such scenarios, we can utilize higher order CBFs \cite{xiao2019control}, a general form of exponential CBFs \cite{nguyen2016exponential}.
\begin{definition}[High Order Control Barrier Function \cite{xiao2019control}]
For a differentiable function $h : X \to \mathbb{R}$ with IRD = $m$, consider a sequence of functions ${\phi_i : X \to \mathbb{R}, i = 1,2, \ldots , m } $, as
\begin{align}
    \phi_0 (x) := h(x), \phi_i (x) := \dot{\phi}_{i-1} (x) + \alpha_i ( {\phi}_{i-1} (x) ),
\end{align}
where ${\alpha_i \in \mathcal{K}_\infty}$ is a ${(m-i)^{th}}$ order differentiable function. We also define the sets ${\mathcal{C}_i := \{ x \in X : {\phi}_{i-1} (x) \geq 0 \}}$. Then, $h(x)$ is a \textit{high order control barrier function (HOCBF)} for system (\ref{system}) on ${\mathcal{C}_1 \intersect \mathcal{C}_2 \intersect \ldots \mathcal{C}_m }$ if there exists ${\alpha_m \in \mathcal{K}_{\infty}}$ and ${\alpha_i \in \mathcal{K}_\infty}$ such that ${\forall x \in {\mathcal{C}_1 \intersect \mathcal{C}_2 \intersect \ldots \mathcal{C}_m }}$:
\begin{align}
\label{hocbf}
   \nonumber
   \sup_{u \in U}  \Big ( {{h}^m(x, u)}  =
   L_f^m h(x) &+ L_g L_f^{m-1} h(x) u + 
   \mathcal{O}(h(x)) \Big )  \\ & \geq -\alpha_m (h_{m-1}(x)) ,
\end{align}
where ${{h}^m(x, u)}$ is the ${m^{th}}$ time derivative of $h(x)$, ${L_f^m h(x) \triangleq \derp{L_f^{m-1} h(x)}{x} f(x)}$, ${L_g L_f^{m-1} h(x) \!\triangleq\! \derp{L_f^{m-1} h(x)}{x} g(x)}$, and ${\mathcal{O}(h(x))\triangleq \sum_{i= 1}^{m-1} L_f^i ( \alpha_{m-i} \circ  {\phi}_{m-i-1} ) }$.
\end{definition}
Given an HOCBF $h(x)$ with ${\alpha_m \in \mathcal{K}_{\infty}}$ and ${\alpha_i \in \mathcal{K}_\infty}$ for system (\ref{system}), the set of safe controllers given as
\begin{equation}
\label{eq:CBFkx} 
    K_\text{HOCBF}(x) \triangleq \{ u \in U  \big |  {{h}^m(x, u)} \geq - \alpha_m (h^{m-1}(x)) \}.
\end{equation}

\section{Problem Statement and Motivation}
\label{sec:pro}
Real-world safety-critical control systems suffer from not completely modelled uncertainties that may deteriorate the safety guarantees of controllers that are designed considering nominal models. We must consider the actual models when synthesizing safe controllers to address this robustness problem, as exemplified in the following example.  

\begin{example} \label{ex:setup}
Consider an adaptive cruise control (ACC), whose controlled vehicle dynamics take the form (\ref{sysun}):
\begin{equation}
\label{ACCdyn}
\underbrace{ \begin{bmatrix} 
\dot{v}_f \\ 
\dot{D}
\end{bmatrix} }_{\dot{x}}
=
\underbrace{ \begin{bmatrix}
\frac{-F_r(v_f)}{M} \\ 
v_l - v_f
\end{bmatrix} }_{\hat{f}(x)}
+ 
\underbrace{ \begin{bmatrix}
\frac{1}{M} \\ 
0
\end{bmatrix} }_{\hat{g}(x)}
u
+ 
\overbrace{
\underbrace{ \begin{bmatrix}
\delta_{F_r,M} \\ 
0
\end{bmatrix} }_{\Delta f(x)}
+ 
\underbrace{ \begin{bmatrix}
\delta_{M^{-1}} \\ 
0
\end{bmatrix} }_{\Delta g(x)} u}^{\Delta(x, u)},
\end{equation}
where ${v_l ~ [m/s]}$ and ${v_f ~ [m/s]}$ are respectively the velocities of the lead car and following (controlled) car, ${D ~ [m]}$ is the distance between the lead and following cars, ${F_r = f_0 + f_1 v_f + f_2 v_f^2 ~ [N]}$ is the aerodynamic drag, ${M ~ [kg]}$ is the following car's mass, and ${\delta_{F_r,M},~ \delta_{{M}^{-1}}}$ are parametric uncertainties (for example ${1/M := \hat{M}^{-1} \pm \delta_{{M}^{-1}}}$) about ${{-F_r(v_f)}/{M}, 1/M }$.

For safety, the following car must maintain a safe distance behind the lead car: ${D \geq v_f \tau_d }$, where $\tau_d$ is a time interval. Therefore, we choose a CBF: ${h(x) \triangleq D - v_f \tau_d}$, which encodes a safe stopping distance. This CBF choice results in 
\begin{align}
    L_{\hat{f}} h(x) &= \tau_d \frac{F_r(v_f)}{M} + v_l-v_f , \\
    L_{\hat{g}} h(x) &=-\frac{\tau_d}{M} , \\
    \frac{\partial h}{\partial x} \Delta(x, u) &= -\tau_d(\delta_{F_r,M} + \delta_{M^{-1}}~u ) . 
\end{align}
The CBF ${h(x)}$ has IRD of 1, and DRD of 1, which satisfies the matching condition.
Thus, we choose a CBF following Definition~\ref{def:cbf} by selecting ${-\alpha (h(x)) = -\alpha h(x),~\alpha \in \mathbb{R}^+}$ for simplicity. Suppose $h(x)$ is a valid CBF for the nominal system (without uncertainty ${\Delta(x, u)}$):  
\begin{align}
\label{eq:cbfacc}
   \tau_d \frac{F_r(v_f)}{M} + v_l-v_f -\frac{\tau_d}{M} u  \geq - \alpha h(x) .
\end{align}
Then, we use the CBF constraint \eqref{eq:cbf_def} for the uncertain model:
\begin{align}
\begin{split}
\label{eq:cbfunacc}
   \overbrace{\tau_d \frac{F_r(v_f)}{M} + v_l-v_f -\frac{\tau_d}{M} u + \alpha h(x)}^{\geq 0} \geq \\
   \tau_d(\delta_{F_r,M} + \delta_{M^{-1}}~u )  
\end{split}
\end{align}
The right-hand side of \eqref{eq:cbfunacc} is the effect of the uncertainty on the CBF constraint. If it is greater than 0, (e.g.: if the real vehicle mass is less than the nominal mass and if $u>0$, then the term ${\delta_{M^{-1}} u}$ will be positive), safety may be violated.
\end{example}

To remedy this problem, we propose to estimate the uncertainty with a quantified estimation error bound and incorporate the effects into the CBF constraints. 

\section{Uncertainty Estimator}
\label{sec:ue}
To estimate the uncertainty ${\Delta(x, u)}$ defined in \eqref{sysun}, \eqref{delta} which we denote as $\hat{\Delta}$, we propose a novel uncertainty estimator that has the structure:
\begin{align}
    \label{eq:bhat}
    \hat{\Delta}(t) &= {\Lambda} x - \xi(t), \\
    \dot{\xi}(t) &= \Lambda \left(\hat{f}(x) + \hat{g}(x) u + \hat{\Delta}(x,\xi(t))\right),  
    \label{eq:xidot}
\end{align}
where ${\xi \in \mathbb{R}^{n}}$ is an auxiliary state vector and ${0 \prec \Lambda \in \mathbb{R}^{n \times n}}$ is a diagonal positive definite estimator design matrix, i.e., ${\Lambda \triangleq \diag(\lambda_1, \ldots, \lambda_n)}$. Without loss of generality,  the initial values of ${\hat{\Delta}(t)}$ and ${\xi(t)}$ are set to zero, i.e., ${\hat{\Delta}(0) = {\bf 0}}$, ${\xi(0) = \Lambda x(0)}$. Note that this estimator is an extension of the disturbance observer studied in \cite{dacs2022robust, alan2022disturbance} to multiple-input multiple-output (MIMO) systems.

The following Lemma characterizes the \textit{input-to-state stability} (ISS) of the estimation error dynamics around zero; it implies that the estimation error is bounded.
\begin{lemma}
Consider uncertain system \eqref{sysun} with a continuously differentiable function $\Delta(x, u)$ that satisfies Assumption~\ref{as1} with a Lipschitz constant ${\delta_L}$, upper bound ${\delta_b}$, and the uncertainty estimator \eqref{eq:bhat}-\eqref{eq:xidot} with a diagonal and positive definite matrix ${\Lambda \in \mathbb{R}^{n \times n}}$. The uncertainty estimator's error dynamics, ${\dot{e}(t) = \dot{{\Delta}}(t) - \dot{\hat{\Delta}}(t)}$, are ISS around ${e = 0}$. \label{lem:Velyap}
\end{lemma}
\begin{proof}
We define the uncertainty estimation error as 
\begin{align}
    e(t) = {\Delta}(t) - \hat{\Delta}(t) , ~ e(t_0) = e_0 
    \label{eq:error}
\end{align}
From \eqref{sysun}, \eqref{delta}, \eqref{eq:bhat} and \eqref{eq:xidot}, we have
\begin{align}
    \dot{e}(t) &= \dot{{\Delta}}(t) - \Lambda e(t) , ~ e(t_0) = e_0.
    \label{eq:dote}
\end{align}
Consider a candidate Lyapunov function ${V_e(e) \triangleq \frac{1}{2} e^T e}$. Taking a time derivative of ${V_e}$ along the trajectory of \eqref{eq:dote} yields
\begin{align}
\begin{split}
    \dot{V} &= e^T ( \dot{{\Delta}} - \Lambda e )
     \leq - e^T \Lambda e + \| e^T \| \underbrace{\| \dot{{\Delta}}\|}_{\leq \delta_L} \\
    & \leq - e^T \Lambda e + \| e^T \| \delta_L  \leq - \lambda_{min} \| e 
 \| ^2 + \| e^T \| \delta_L \label{eq:dotV_intermediate}  
 \end{split}
\end{align}
The last inequality arises from the fact that since matrix $\Lambda$ is a real symmetric positive definite, Rayleigh's inequality holds: ${\lambda_{min} \| e
 \| ^2 \leq e^T \Lambda e \leq \lambda_{max}\|e\|^2}$, where ${\lambda_{min}, \lambda_{max}}$ are the minimum and maximum eigenvalues of $\Lambda$. Introduce an auxiliary non-negative constant, $\left (  \sqrt{\lambda_{min}} \| e
 \| -  \dfrac{\delta_L}{\sqrt{\lambda_{min}}} \right )^2 \geq 0$. We replace ${\| e^T \| \delta_L}$ with the following upper bound:
 \begin{align}
 \| e^T \| \delta_L \leq \dfrac{\lambda_{min}}{2} \| e
 \| ^2 + \dfrac{\delta_L^2}{2 \lambda_{min} }. \label{eq:dotV_intermediate2}
\end{align}
Finally, substituting \eqref{eq:dotV_intermediate2} into \eqref{eq:dotV_intermediate} yields
 \begin{align}
  \dot{V}(e, \dot{e})  \leq - \dfrac{\lambda_{min}}{2} \| e
 \| ^2 + \dfrac{\delta_L^2}{2 \lambda_{min} } ,
 \label{eq:vedot}
\end{align}
which ensures ISS of the observer error dynamics near ${e = 0}$: 
\begin{align}
\dot{V}(e, \dot{e})  \leq -2 \mu_e \| e
 \| ^2 + \gamma(\| \delta_L \|),
 \label{eq:er_iss}
\end{align}
with ${ \mu_e \triangleq \frac{\lambda_{min}}{4}}$, and a class-${\mathcal{K}}$ function ${\gamma(\| \delta_L \|) \triangleq \frac{\delta_L^2}{2 \lambda_{min} }}$.
\end{proof}
From Lemma~\ref{lem:Velyap}, the ISS property ensures the boundedness of error dynamics ${\|e(t)\|}$. This motivates us to derive an explicit tight bound, which results from applying Lemma~\ref{lem:expA} to the estimation error and output. 
\begin{lemma}(Adapted from \cite{nechepurenko2002})
\label{lem:expA}
Consider a Hurwitz matrix  ${{\rm A} \in \mathbb{R}^{n \times n}}$. The norm ${\|{\rm e}^{{\rm A} t}\|}$ is bounded by 
\begin{equation}
    \label{eq:ebound}
    {\|{\rm e}^{{\rm A} t}\|}
     \leq  \left ( \sqrt{\| P^{-1}\| \| P \|}\right )   {\rm e}^{ -\left ( \dfrac{t}{2 \| P \|} \right ) } ,
\end{equation}
where ${P}$ is the Hermitian matrix solution that solves the Lyapunov equation: ${P {\rm A} + {\rm A}^T P  = -I }$.
\end{lemma}

\begin{lemma}
\label{lem:Vb}
Under the same assumptions in Lemma~\ref{lem:Velyap}, the following bounds hold for estimation error ${e(t) \!=\! {\Delta}(t) \!-\! \hat{\Delta}(t)}$ and estimated uncertainty ${\hat{\Delta}(t)}$, respectively.
\begin{align}
    \|e(t)\| &\leq \mathcal{P}   \left ( \delta_b - 2 {\delta_L}{\| P \|} \right ) {\rm e}^{ -\left ( \dfrac{t}{2 \| P \|} \right ) }  + 2 \mathcal{P} {\| P \|} \delta_L   \\
    \|\hat{\Delta}(t)\| &\leq   2 \mathcal{P} \delta_b \| \Lambda \| { \| P \|}     \Bigg ( 1 -  {\rm e}^{ -\left ( \dfrac{t}{2 \| P \|} \right ) }  \Bigg  )
\end{align}
where ${P {\rm \Lambda} + {\rm \Lambda}^{T}P  = -I }$, and ${  \mathcal{P} \triangleq \sqrt{\| P^{-1}\| \| P \|}}$.
\end{lemma}
\begin{proof}
From \eqref{sysun}, \eqref{delta}, \eqref{eq:bhat}, \eqref{eq:xidot} and \eqref{eq:error}, we have
\begin{align}
    \dot{\hat{\Delta}}(t) &= \Lambda (\Delta(x, u) - \hat{\Delta}(t)).
    \label{dotb}
\end{align}
Under Assumption~\ref{as1}, we have ${\| \Delta(x, u) \| \leq \delta_b}$, and ${\|\Dot{\Delta}(t) \|\leq \delta_L}$. Since we set ${\hat{\Delta}(0) = {\bf 0}}$, the initial error will be bounded by ${\| e_0 \| \leq \delta_b}$. 
The solution to ${e(t)}$ in \eqref{eq:dote} with the state transition matrix ${ {\rm e}^{-\Lambda t}}$ and initial condition ${e_0}$
is given by
\begin{align*}
    {e}(t) &=  {\rm e}^{-\Lambda t} e_0 + \int_{0}^{t}  {\rm e}^{-\Lambda (t-\tau)} \dot{{\Delta}}(\tau) d \tau   .
\end{align*}
Taking norms and integrating the right-hand side yield
\begin{align*}
     \| {e}(t) \| &\leq \big \|{\rm e}^{-\Lambda t} \big \|  \underbrace{\|  e_0 \|}_{\leq \delta_b} + \int_{0}^{t} \big \| {\rm e}^{-\Lambda (t-\tau)} \big \| \underbrace{\| \dot{{\Delta}}(\tau) \|}_{\leq \delta_L} d \tau \\
     &\leq \delta_b \mathcal{P}  {\rm e}^{ -\left ( \dfrac{t}{2 \| P \|} \right ) } + \delta_L \int_{0}^{t} \mathcal{P}  {\rm e}^{ -\left ( \dfrac{t - \tau}{2 \| P \|} \right ) } d \tau \\
     &= \mathcal{P}\left( \left ( \delta_b - 2 {\delta_L}{\| P \|} \right ) {\rm e}^{ -\left ( \dfrac{t}{2 \| P \|} \right ) }  + 2 {\| P \|} \delta_L  \right)  ,
\end{align*}
which yields the first theorem statement. The second inequality follows from Lemma~\ref{lem:expA}. Note that the derived bound is in the form of an input-to-state stable system requirement:
\begin{align*}
    \|e(t)\| &\leq \beta(\|e_0\|, t) + \gamma(\| \delta_L \|) ,
\end{align*}
where ${ \beta \in \mathcal{KL}, \gamma \in \mathcal{K}}$.

To derive the second statement of the theorem, we consider the integration of \eqref{dotb} with ${\hat{\Delta}(0) = {\bf 0}}$:
\begin{align*}
    {\hat{\Delta}}(t) &=  {\rm e}^{-\Lambda t} \underbrace{\hat{\Delta}(0)}_{0} + \int_{0}^{t}  {\rm e}^{-\Lambda (t-\tau)} \Lambda {{\Delta}}(\tau) d \tau   \\
    &= \int_{0}^{t}  {\rm e}^{-\Lambda (t-\tau)} \Lambda {{\Delta}}(\tau) d \tau  
\end{align*}
Taking norms and integrating the right-hand side yield
\begin{align*}
    \| {\hat{\Delta}}(t) \| & \leq \delta_b \| \Lambda \|  \int_{0}^{t} \Big \| {\rm e}^{-\Lambda (t-\tau)}  \Big \| d \tau  \\
    & \leq \mathcal{P} \delta_b \| \Lambda \| \int_{0}^{t}  {\rm e}^{ -\left ( \dfrac{t - \tau}{2 \| P \|} \right ) } d \tau \\
    &=  2 \mathcal{P} \delta_b \| \Lambda \| { \| P \|}   \Bigg ( 1 -  {\rm e}^{ -\left ( \dfrac{t}{2 \| P \|} \right ) }  \Bigg  )
\end{align*}
which is the second statement of the theorem.
\end{proof}

\begin{remark}
The first item of Lemma~\ref{lem:Vb} states the \textit{input-to-state stability} of the observer error dynamics around ${e=0}$, and it consists of transient and steady-state parts. The second item shows that if the uncertainty and its derivative are bounded, then the proposed estimator results in a bounded output.    
\end{remark}

Lemma~\ref{lem:Velyap} presents a property of the estimator's  time derivative, leveraging Lyapunov functions, while Lemma~\ref{lem:Vb} presents the time-dependent convergence of the estimator. The following section introduces two methods to assert safety with the estimator \eqref{eq:bhat}-\eqref{eq:xidot} using these results.

\section{Main Results}
To guarantee the robustness of CBFs, we must incorporate the uncertainty ${\Delta(x,u)}$, but it is not directly accessible in actual implementations. A possible solution to this issue is to replace ${\Delta(x,u)}$ with the estimated uncertainty term ${\hat{\Delta}(t)}$ and the lower bound of the associated estimation error ${\|{e}(t)\|}$. 

\subsection{Method 1}
This section uses the proposed estimator's active uncertainty attenuation capability to reject uncertainty via input augmentation. We also show that this method provides a performance improvement on top of guaranteed robust safety.
 
The \textit{disturbance decoupling problem \cite{sastry2013nonlinear} (Theorem, 9.20)} states that finding a state feedback controller for system \eqref{sysun}, such that the disturbance has no effect on the system output is possible iff IRD = DRD. Therefore, most disturbance observer-based robust control methods apply only to the compensation of matched disturbance inputs. While a few studies have proposed nonlinear disturbance observers that eliminate system output disturbances, they ignore the transient dynamics under disturbance effects \cite{yang2012} or consider specific classes of nonlinear systems, e.g., lower triangular systems. Therefore, for simplicity, we assume {IRD = DRD = 1} and uncertainty satisfies the matching condition. 

We utilize Lemma~\ref{lem:Velyap} in this section. 
Suppose that there exists a matrix ${Q \in \mathbb{R}^{m \times n}}$ that satisfies 
\begin{equation}
\hat{g}(x) Q(x) \Delta(x, u) = \Delta(x, u), ~ \forall t \geq 0, 
\end{equation}
which matches the lumped uncertainty into the input channel:
\begin{align}
\begin{split}
\label{eq:distd}
    \dot{x} & = \hat{f}(x) + \hat{g}(x) u +\Delta(x, u) \\
   &\equiv \hat{f}(x) + \hat{g}(x) (u + Q(x) \Delta(x, u)).
\end{split}
\end{align}
Finding a valid matrix $Q(x)$ for equality (\ref{eq:distd}) is equivalent to decoupling the disturbance from the system states, which is a well-studied problem, \cite{isidori1985}. However, to our knowledge, no formulation exists for general MIMO nonlinear systems. 

While a general formula is beyond the scope of this study, we define a useful matrix ${Q(x)}$.
If ${rank (\hat{g}(x)) = m}$, for all ${x \in X}$, and its pseudo-inverse ${(\hat{g}(x)^{T} \hat{g} (x))^{-1} \hat{g}^{T} (x)}$ exists for all ${x \in X}$, then 
\begin{equation}
    \label{eq:qx}
    Q(x) \triangleq (\hat{g}(x)^{T} \hat{g} (x))^{-1} \hat{g}^{T} (x).
\end{equation}

We revisit the ACC problem of Example~\ref{ex:setup} to demonstrate how to find a ${Q(x)}$ matrix:
\begin{align}
\begin{split}
\underbrace{ \begin{bmatrix}
 \frac{1}{M} \\ 
0
\end{bmatrix} }_{\hat{g}(x)} 
Q(x) &
\underbrace{ \begin{bmatrix}
\delta_{F_r,M} + \delta_{M^{-1}} u \\ 
0
\end{bmatrix}}_{\Delta(x, u)}
= \begin{bmatrix}
\delta_{F_r,M} + \delta_{M^{-1}} u \\ 
0
\end{bmatrix} , \\
\implies Q(x) &= \left ( \begin{bmatrix}
\dfrac{1}{M} & 0
\end{bmatrix} 
 \begin{bmatrix}
\dfrac{1}{M} \\ 
0
\end{bmatrix}
\right )^{-1}   \begin{bmatrix}
\dfrac{1}{M} & 0
\end{bmatrix}  \\
&= 
\begin{bmatrix}
M & 0 
\end{bmatrix}.
\end{split}
\end{align}

Since Lemma~\ref{lem:Velyap} holds for the uncertainty estimation error ${e(t)}$ for the proposed estimator with Lyapunov function ${V_e(e) = \frac{1}{2} e^Te}$, we have ${\Dot{V}_e(e, \dot{e})  \leq -2 \mu_e V_e + \gamma(\| \delta_L \|)}$, which is given in \eqref{eq:er_iss}. Using this property, we incorporate the uncertainty estimator \eqref{eq:bhat}, \eqref{eq:xidot} into the CBF setup with the augmentation of the given CBF and the control input. 

To compensate for the uncertainty, we incorporate the uncertainty estimator into the feedback control law as
\begin{align}
    \label{eq:u_hat}
    u = \tilde{u}  - \underbrace{Q(x) \hat{\Delta}}_{u_{\hat{\Delta}}}, 
\end{align}
where ${\tilde{u} = \mathbf{k} (x)}$ is the feedback controller and parameter ${u_{\hat{\Delta}}}$ is for the uncertainty attenuation (See Fig.~\ref{fig:scheme} top). Then, we obtain the closed-loop dynamics:
\begin{align}
\begin{split}
\label{syscon}
    \dot{x}  &= \hat{f}(x) + \hat{g}(x) \left(\tilde{u}  - u_{\hat{\Delta}}\right) + \Delta(x, u) \\
     & = \hat{f}(x) + \hat{g}(x) \tilde{u} + \underbrace{\Delta(t) - \hat{\Delta}(t)}_{e(t)}, 
\end{split}
\end{align}
which depends on the uncertainty estimation error. 

Then, inspired by \cite{molnar2021model}, we modify the CBF and its associated safe set to provide robustness against estimation error $e(t)$:
\begin{equation}
\label{hecbf}
    h_V (x, e)  = h(x) - \sigma_V V_e(e) , ~ \sigma_V \in \mathbb{R}^+ . 
\end{equation}
\begin{equation}
\label{hexf}
    \C_V := \left\{ x\in X  ~|~ h_V(x, e)\geq 0 \right\} , 
\end{equation}
which is a subset of the safe set ${\C := \left\{ x\in X  ~|~ h(x)\geq 0 \right\}}$ since ${V_e(e)}$ is a Lyapunov function.  The following theorem relates the controllers designed for the safety of the nominal system to the safety of the uncertain system by considering the estimator dynamics.
Note for the following theorem; we use ${-\alpha (h(x)) = -\alpha h(x),~\alpha \in \mathbb{R}^+}$ for simplicity, and we drop the arguments of functions for convenience in the proof.
\begin{theorem}
\label{theo:main}
Consider the uncertain system \eqref{sysun}, a valid CBF $h(x)$ defining the set $\C$ as ${\C := \left\{ x\in X  ~|~ h(x)\geq 0 \right\}}$ such that ${\frac{\partial h}{\partial x} \neq 0}$, ${\forall x \in \C}$, with a continuously differentiable uncertainty function ${\Delta(x,u)}$ defined by \eqref{delta} that satisfies Assumption~\ref{as1} with a Lipschitz constant ${\delta_L}$, upper bound ${\delta_b}$, the uncertainty estimator that satisfies \eqref{eq:er_iss} and ${\hat{\Delta}(0) = {\bf 0}}$, and a Lipschitz continuous controller ${u}$. If there exist a constant ${\mu_h \in \mathbb{R}^+}$ such that 
\begin{equation}
    \label{eq:d_muh}
    \mathcal{D} \triangleq 4\sigma_{V}\mu _e - 2\sigma_V \mu_h > 0,
\end{equation}
and the following condition holds
\begin{equation}
\label{heineq}
    L_{\hat{f}} h(x) + L_{\hat{g}} h(x) \tilde{u}    \geq -\mu_h h(x) + \dfrac{ \norm{\dfrac{\partial h }{\partial x}} ^2 }{\mathcal{D} } + \sigma_V \gamma(\| \delta_L \|) ,
\end{equation}
then ${x_0 \in \C \implies x(t) \in \C}$.
\end{theorem}
\begin{proof}
Our goal is to obtain ${\dot{h}_V \geq  -\mu_h h_V, ~ \forall t \geq 0}$, which implies ${\dot{h} \geq  -\mu_h h, ~ \forall t \geq 0}$. The time derivative of ${h_V(x, e)}$ given in \eqref{hecbf} satisfies:
\begin{align*}
  \dot{h}_V &=  \dot{h} - \sigma_V \dot{V}_e \\ &\geq \dot{h} - \sigma_V (-\mu_e e^2 + \gamma(\| \delta_L \|) ) \\
    &= L_{\hat{f}} h(x) + L_{\hat{g}} h(x) {u} + \frac{\partial h}{\partial x} \Delta + \sigma_V \mu_e e^2 - \sigma_V \gamma(\| \delta_L \|) \\
     &=  L_{\hat{f}} h(x) + L_{\hat{g}} h(x) \tilde{u} + \frac{\partial h}{\partial x} \underbrace{({\Delta} - \hat{\Delta})}_{:=e} - \sigma_V \gamma(\| \delta_L \|) \\ &~~+ \left ( \sigma_V \mu_e -  \frac{\sigma_V\mu_h}{2} \right ) e^2 + \frac{\sigma_V\mu_h}{2} e^2   \\
     &= L_{\hat{f}} h(x) + L_{\hat{g}} h(x) \tilde{u} - \sigma_V \gamma(\| \delta_L \|)     \\
     &~~+ \left ( \frac{ \sqrt{\mathcal{D}}e}{2} +  \dfrac{ {\frac{\partial h}{\partial x}} }{\sqrt{\mathcal{D}}}  \right )^2 - \dfrac{ \norm{\frac{\partial h}{\partial x}} ^2 }{\mathcal{D}} + \frac{\sigma_V\mu_h}{2} e^2  \\
     &\geq \underbrace{L_{\hat{f}} h(x)\! +\! L_{\hat{g}} h(x) \tilde{u}  \! -\! \dfrac{ \norm{\frac{\partial h}{\partial x}} ^2 }{\mathcal{D}} \!-\! \sigma_V \gamma(\| \delta_L \|)}_{\geq -\mu_h h(x)} \! + \frac{\sigma_V \mu_h e^2}{2} \\
     &\implies \dot{h}_V \geq -\mu_h h(x) +\frac{\sigma_V \mu_h e^2}{2} = -\mu_h h_V
\end{align*}
This leads to
${h(x(t)) \geq 0}$, that is, ${x(t) \in \C}$, ${\forall t \geq 0}$.
\end{proof}
The second proof line uses a lower bound for ${V_e}$. In the fourth line, we replaced $u$ with ${\tilde{u} - u_{\hat{\Delta}}}$ (see \eqref{eq:u_hat}) to obtain a term that explicitly depends on ${e}$.  In the next line, we use a square to define a new lower bound. Finally, using the condition in the theorem, we obtain the main statement.

Then, given an $h(x)$, $\mu_h \in \mathcal{K}_{\infty, e}$, ${\sigma_V}$ for system (\ref{sysun}), we define the set of robust safe controllers as
\begin{equation}
\label{CBFkx} 
    K_\text{UE}(x) \triangleq \left\{ \tilde{u} \in U  \big |  L_{\hat{f}} h(x) + L_{\hat{g}} h(x) \tilde{u}    \geq \mathcal{S}(x) \right\},
\end{equation}
where ${\mathcal{S}(x) = -\mu_h h(x) + \dfrac{ \norm{\dfrac{\partial h}{\partial x}} ^2 }{\mathcal{D}} + \sigma_V \gamma(\| \delta_L \|)}$.

Finally, we have the following robust CBF-QP:
\begin{argmini*}|s|
{\tilde{u}  \in U }{\|\tilde{u} -\mathbf{k_d}(x) \|^2 }
{\label{robCBF0}}
{\tilde{u} ^*({x})=}
\addConstraint{L_{\hat{f}} h(x) + L_{\hat{g}} h(x) \tilde{u}    \geq \mathcal{S}(x)  }
\end{argmini*}


\subsection{Method 2}
Method 1 composed the estimated signals and the feedback controller. Even though Method 1 is useful for disturbance attenuation, it may not be useful for high dimensional systems as it depends on the existence of matrix ${Q(x)}$. Furthermore, Method 1 is mainly aimed at CBFs with relative degree one.  This section proposes an alternative method to address model uncertainty for HOCBFs by integrating the proposed estimator with CBFs only for the safety condition.

We construct a robust safety constraint using the estimated part of the uncertainty and an associated error bound inspired by \cite{dacs2022robust,cheng2022safe}. Specifically, we utilize the proposed estimator with the main outputs of Lemma~\ref{lem:Vb} to robustify the given CBFs for the nominal systems. We incorporate the uncertainty estimator into the CBF constraint \eqref{eq:hdot} by plugging in for ${\Delta(x, u)}$:
\begin{align}
    \label{eq:Khat_CBF}
    \!\!\!{\dot{h}(x, u)}\!  =\!
   \dfrac{\partial h}{\partial x} ( \hat{f}(x)\! +\! \hat{g} (x) u ) \! +\! \frac{\partial h}{\partial x} \!\underbrace{ {\Delta(x, u)}}_{=  \hat{\Delta} + e } \!\geq\! -\alpha (h(x)) .
\end{align}
Due to uncertainty in the estimation error, we can not directly replace ${\Delta(x, u)}$ in \eqref{eq:Khat_CBF} with ${ \hat{\Delta} + e}$, which requires the exact information about the estimation error ${e}$ and the uncertainty. Therefore, to compensate for the uncertainty in \eqref{eq:Khat_CBF}, we use the derived time-dependent estimator bounds.
\begin{theorem}
\label{theo:met2}
Under the same assumptions of {Theorem~\ref{theo:main}}, if ${\frac{\partial h}{\partial x} \neq 0}$, ${\forall x \in \C}$ and there exists ${\alpha \in \mathcal{K}_{\infty, e}}$ such that ${\forall x \in \mathcal{C}}$
\begin{align} \label{eq:thm2_set}
\!\!\!\sup_{u \in U}  \!\Bigg (\! \dfrac{\partial h}{\partial x} ( \hat{f}(x) \!+\! \hat{g} (x) u  \!+\!  \hat{\Delta} ) \!-\! \norm{\frac{\partial h}{\partial x}} \!\|e\|  \!\Bigg ) \!\geq\! -\alpha (h(x)),
\end{align}
then ${x_0 \in \C \implies x(t) \in \C}$.
\end{theorem}
\begin{proof} A lower bound for the time derivative of $h(x)$ can be derived using \eqref{eq:error}, \eqref{eq:hdot} and \eqref{eq:Khat_CBF} as
\begin{align}
\label{pr1}
   \!\!\!{\dot{h}(x, u)}  &\! =\!  \dfrac{\partial h}{\partial x} ( \hat{f}(x) + \hat{g} (x) u ) + \frac{\partial h}{\partial x} ({\hat{\Delta} + e}) \\
   & \!\geq\! L_{\hat{f}} h(x) + L_{\hat{g}} h(x) u + \frac{\partial h}{\partial x} {\hat{\Delta} - {\norm{\frac{\partial h}{\partial x}} \|e\| }},
   \label{eq:the2}
\end{align}
where the term ${\norm{\frac{\partial h}{\partial x}} \|e(t)\| \geq 0}$ is to provide a lower bound against the effect of the estimation error $e(t)$ on the CBF constraint, which results in ${x(t) \in \C}$, ${\forall t \geq 0}$. 
\end{proof}

Method 2 also applies to HOCBFs, as supported by the following corollary.   
\begin{corollary}
\label{theo:met3}
Consider a sufficiently differentiable and valid HOCBF ${h(x)}$ for the uncertain model \eqref{sysun} with a sequence of functions ${\phi_i : X \to \mathbb{R}, i = 1,2, \ldots , m } $, and assume that {IRD = DRD = $m$}. Also consider a continuously differentiable function ${\Delta(x,u)}$ defined by \eqref{delta} that satisfies Assumption~\ref{as1} with a Lipschitz constant ${\delta_L}$, upper bound ${\delta_b}$, the uncertainty estimator that satisfies \eqref{eq:er_iss} and ${\hat{\Delta}(0) = {\bf 0}}$, and a Lipschitz continuous controller ${u}$. Then, ${h(x)}$ is a HOCBF if there exists ${\alpha_m \in \mathcal{K}_{\infty}}$ and ${\alpha_i \in \mathcal{K}_\infty}$ such that ${\forall x \in {\mathcal{C}_1 \intersect \mathcal{C}_2 \intersect \ldots \mathcal{C}_m }}$: 
\begin{align}
\begin{split}
\label{eq:the_hocbf}
   \sup_{u \in U}  \Big ( L_{\hat{f}}^m h(x) + L_{\hat{g}} L_{\hat{f}}^{m-1} h(x) u + \derp{L_{\hat{f}}^{m-1} h(x)}{x} \hat{\Delta} \\
  - \Bigg \|{\frac{L_{\hat{f}}^{m\!-\!1}\! h(x)}{\partial x} } \Bigg \| \|e(t) \| \!+\!
   \mathcal{O}(h(x)) \Big )
    \!\geq\! -\alpha_m (h_{m-1}(x)) ,
\end{split}
\end{align}
\end{corollary}
\begin{proof} 
Follows from Theorem~\ref{theo:met2} by replacing the uncertainty term with the associated estimation and error bound.
\end{proof}

\subsection{Discussion}
One of the differences between the proposed methods is that Method 2 explicitly depends on the estimation error, whereas in Method 1, we consider a subset of the original safe set with a design parameter $\sigma_V$. As can be understood from \eqref{hecbf}, choosing a larger $\sigma_V$ may lead to a conservative safe set. Therefore, the tuning process for this parameter is a crucial step for Method 1. On the other hand, there is no additional tuning parameter for Method 2. To take this advantage and leverages the idea of uncertainty compensation, we can incorporate the estimation error dynamics into the CBF constraint for the closed-loop system \eqref{syscon} with the following corollary.  
\begin{corollary}
\label{corol:met3}
Under the same assumptions of {Theorem~\ref{theo:main}}, if ${\frac{\partial h}{\partial x} \neq 0}$ ${\forall x \in \C}$ and there exists ${\alpha \in \mathcal{K}_{\infty, e}}$ such that ${\forall x \in \mathcal{C}}$
\begin{align} \label{eq:thm4_set}
\!\!\!\sup_{\tilde{u} \in U}  \!\Bigg (\! \dfrac{\partial h}{\partial x} ( \hat{f}(x) \!+\! \hat{g} (x) \tilde{u} )  \!-\! \norm{\frac{\partial h}{\partial x}} \!\|e\|  \!\Bigg ) \!\geq\! -\alpha (h(x)),
\end{align}
then ${x_0 \in \C \implies x(t) \in \C}$.
\end{corollary}
\begin{proof} 
Since the term ${(\!-\! \norm{\frac{\partial h}{\partial x}} \!\|e\|   )}$ is a lower bound for the effect of the uncertainty estimation error on the time derivative of CBF, \eqref{eq:thm4_set} results in ${x(t) \in \C}$, ${\forall t \geq 0}$. 
\end{proof}
Using Corollary~\ref{corol:met3}, we define an alternative robust {CBF-QP} for Method 1:
\begin{argmini*}|s|
{\tilde{u}  \in U}{\|\tilde{u} -\mathbf{k_d}(x) \|^2 }
{\label{robCBF1}}
{\tilde{u} ^*({x})\!=\!}
\addConstraint{\dfrac{\partial h}{\partial x} ( \!\hat{f}(x)\! \!+\! \hat{g} (x) \tilde{u} )  \!-\! \!\norm{\frac{\partial h}{\partial x}}\! \!\|e\|\!     \!\geq\! \!-\!\alpha (\!h(x)\!) }
\end{argmini*}
which explicitly depends on the error dynamics of the uncertainty estimator.

\section{Simulations}
 The parameters of the simulation are taken from \cite{dacs2022robust} (Table 1).
\begin{figure*}[t]
	\centering
 \vspace{2mm}
	\includegraphics[width=1.0\linewidth]{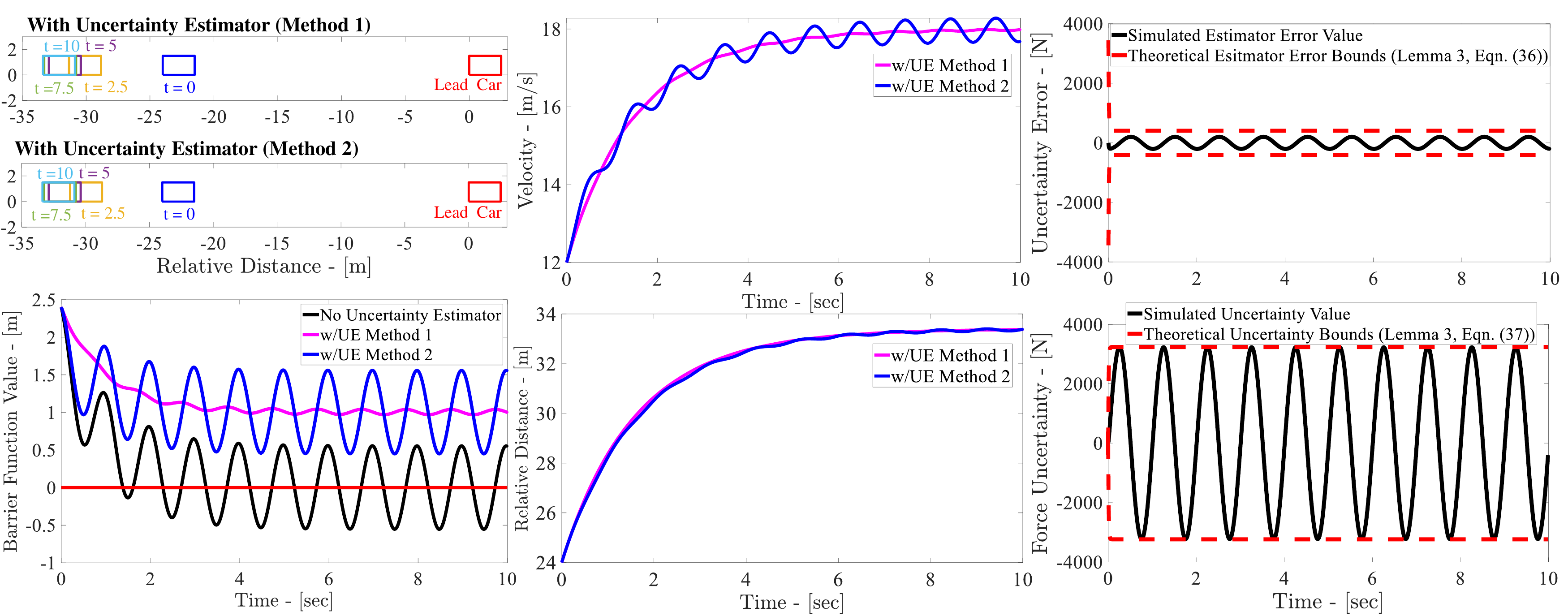}
	\caption{Simulations for the ACC Example with uncertainty. (Left) Evaluation of CBF h(x). The proposed controller maintains safety in the presence of unmodelled dynamics.
    (Middle) The closed-loop system performance is improved by disturbance rejection (Method 1).
    (Right) The estimation error and output satisfy the theoretical bounds.} 
	\label{fig:result_acc}
\end{figure*}
This section completes the ACC example with the proposed uncertainty estimator, and uses a multi-rotor example to illustrate Method 2.
\subsection{Adaptive Cruise Control Example}
 The parameters of the simulation are taken from \cite{dacs2022robust} (Table 1).
The closed-loop control performance requires cruise at a constant speed encoded into the QP via control Lyapunov function (CLF) ${V(x) = ( v_f - v_d)^2}$. We augment robust CBF-QP with the related CLF constraint; hence we have the following robust CLF-CBF-QP:
\begin{argmini*}|s|
{{u}  \in \U, \delta_c \in \mathbb{R}}{\|{u} \|^2 + p_c \delta_c^2 }
{\label{acc_QP}}
{\tilde{u} ^*({x})=}
\addConstraint{L_{\hat{f}} h(x) + L_{\hat{g}} h(x) {u}    \geq \mathcal{S}(x)  }
\addConstraint{\dot{V}(x, u)  \leq -\lambda V(x) + \delta_c }
\end{argmini*}
where ${\delta_c}$ is a relaxation variable, ${p_c = 100}$, we pick ${\lambda = 0.7}$. 

We set the first row of ${\Delta(t) = (\delta_{F_r,M} + \delta_{M^{-1}}~u )}$ given in \eqref{ACCdyn}, as ${\Delta(t):= \frac{2sin(2 \pi t)}{M} + \frac{0.2 F_r(v_f)}{M} + \frac{1}{2 M} u }$, which represents both the drag and mass uncertainties and violates the safety as discussed in \eqref{eq:cbfunacc}. To obtain ${\delta_L, \delta_b \in \mathbb{R}^+}$ constants for this uncertainty, we run the simulation for possible scenarios. Then, we have ${\delta_L = 26,~\delta_b = 12}$.  
We choose ${\Lambda = \diag(100, 100)}$ to design the proposed uncertainty with a small estimation error; therefore, ${ \mu_e := \frac{\lambda_{min}}{4} = 25}$. For Method 1, we pick ${\mu_h = 1}$ and ${\sigma_V = 0.1}$, so ${\mathcal{D} \triangleq 4\sigma_{V}\mu _e - 2\sigma_V \mu_h > 0}$.

Fig.~\ref{fig:result_acc} shows the evaluation of CBF $h$, the states for an initial set ${x(0) = [18 \ 12 \ 24]^T}$, and the uncertainty estimation performance of the proposed observer, respectively.
As seen from Fig.~\ref{fig:result_acc}-(Left), the unmodelled dynamics cause safety violations without the estimator. 
Method 2 provides a safety guarantee but does not compensate for the effect of disturbance on the states (${v_f, D}$), as observed from Fig.~\ref{fig:result_acc}-(Middle) Method 1 achieves the disturbance attenuation requirement. 
Finally, it is observed from  Fig.~\ref{fig:result_acc}-(Right) that the proposed uncertainty estimator can effectively estimate actual modelling uncertainties within the quantified bounds.
 
\vspace{-1mm}
\subsection{Multirotor Trajectory Tracking Example}
\vspace{-1mm}
The multirotor states are its global position ${\mathrm{p} = [\mathrm{x},\mathrm{y},\mathrm{z}] \in \mathbb{R}^3}$, and Euler angle parameterization of rotation matrix ${R(\phi,\theta,\psi) \in \text{SO(3)}}$ where ${\phi}$, ${\theta}$, and ${\psi}$ are roll, pitch, and yaw angles. Rotation rates ${p,~q}$, and $r$ are defined with respect to the body-fixed $\mathrm{x}$, $\mathrm{y}$, and $\mathrm{z}$ axes. If $T$ denotes the total axial thrust on the multirotor, we have the following dynamics (nominal) where gravity is compensated to ensure ${f(0) = 0}$: 
\begin{align} \label{eq:pos}
\ddot{p} &= TR(\phi,\theta,\psi)e_3,\qquad  \dot{T} = u_1,\\
   \begin{bmatrix}
       \dot{\phi}\\
       \dot{\theta}\\
       \dot{\psi}
   \end{bmatrix} &= \underbrace{\begin{bmatrix}
       1 & \sin(\phi)\tan(\theta) & \cos(\phi)\tan(\theta)\\
       0 & \cos(\phi) & -\sin(\phi)\\
       0 & \sin(\phi)\cot(\theta) & \cos(\phi)\cot(\theta)
   \end{bmatrix}}_{W(\phi,\theta,\psi)} \begin{bmatrix}
       u_2\\
       u_3 \\
       u_4
   \end{bmatrix}, \label{eq:attitude}
\end{align}
where $\mathrm{m}$ is the vehicle mass, ${e_{3}=[0,0,1]^T}$. ${u = [u_1, u_2, u_3, u_4]}$ are the control inputs where $u_1$ is the rate of change of mass normalized thrust, ${u_2, u_3}$, and $u_4$ are the body rotation rates with respect to the body-fixed $\mathrm{x}$, $\mathrm{y}$, and $\mathrm{z}$ axes.  The system outputs are ${s = [\mathrm{x},\mathrm{y},\mathrm{z},\psi]^{T}}$, and the goal is to track a desired  differential trajectory ${s_d:[t_{0},t_{f}] \to \mathbb{R}^{4}}$. Let ${e_p = p - [x^d,y^d,z^d]^T}$ be the position tracking errors and ${e_{\psi} = \psi - \psi^d}$ be the yaw angle tracking error. To obtain an error dynamics in control-affine form, we perform a time-derivative on \eqref{eq:pos} to obtain the following:
\begin{align} \label{eq:drone_dynamics}
   \frac{d}{dt} \underbrace{\begin{bmatrix}
        e_p\\
        \dot{e}_p\\
        \ddot{e}_p\\
        e_{\psi}
    \end{bmatrix}}_{\eta} =\mathcal{F} \eta +  \underbrace{\begin{bmatrix}
        0_{6\times 4} \\
        I_{4\times 4}
        \end{bmatrix}}_{\mathcal{G}} \underbrace{\underbrace{\begin{bmatrix}
            Re_3 & -TRe_3^{\vee}\\
            0 & e_3^{T}W
        \end{bmatrix}}_{B_{u}} \begin{bmatrix}
            u_1 \\
            u_2 \\
            u_3 \\
            u_4
        \end{bmatrix}}_{v},  
        \vspace{-1mm}
\end{align}
where
${\mathcal{F}=\left[\begin{smallmatrix}
        0_{6\times 3} &I_{6\times 6} & 0_{6\times 1}\\
        0_{4\times 3} &0_{4\times 6} & 0_{4\times 1}
    \end{smallmatrix}\right]}$. 
For simplicity, we assume full-state observability. The output has a vector relative degree ${\vec{\gamma} = [3,3,3,1]}$. We aim to find a family of tracking controllers that would exponentially stabilize the error dynamics \eqref{eq:drone_dynamics} to the origin. Consider a feedback control law in terms of the auxiliary input ${v = K\eta}$, the goal to stabilize the closed-loop system ${\dot{\eta} = (\mathcal{F}+\mathcal{G}K)\eta = A_{cl}\eta}$. Let ${V_{\eta} = \eta^{T}P\eta}$ where the positive-definite matrix $P$ is the solution to the Continuous-Time Algebraic Equation (CARE) ${\mathcal{F}^TP +P \mathcal{F} - P \mathcal{G}\mathcal{G}^T P = -Q}$, where ${Q = I}$ to ensure maximum regions of attraction. Nominally, we obtain a stabilizing controller by solving a standard CLF-QP.

For safety, known static spherical obstacles must be avoided regardless of the reference trajectory. Let ${p_i^{obs} = [x_i^{obs},y_i^{obs},z_i^{obs}]^T}$ and $r_i$ denote the position of the $i^{th}$ static spherical obstacle in global/inertial frame and the radius of the obstacle, respectively. We use the following CBF for obstacle avoidance:${~
    h_i = \|p-p^{obs}_i\| - r_i - r_0,
}$
where $r_0$ is the radius of the multirotor. To match the relative degree, we employ a HOCBF:${~
    h_{e,i}\! =\! L_{f}^2 h_i\! +\! \alpha_{2}L_{f}h_i \!+\! \alpha_{1}h_i,
}$
where ${\alpha_1, \alpha_2\!\in\!\mathbb{R}^+}$ are tunable constants.
\begin{figure}[t]
	\centering
	\includegraphics[width=0.9\linewidth]{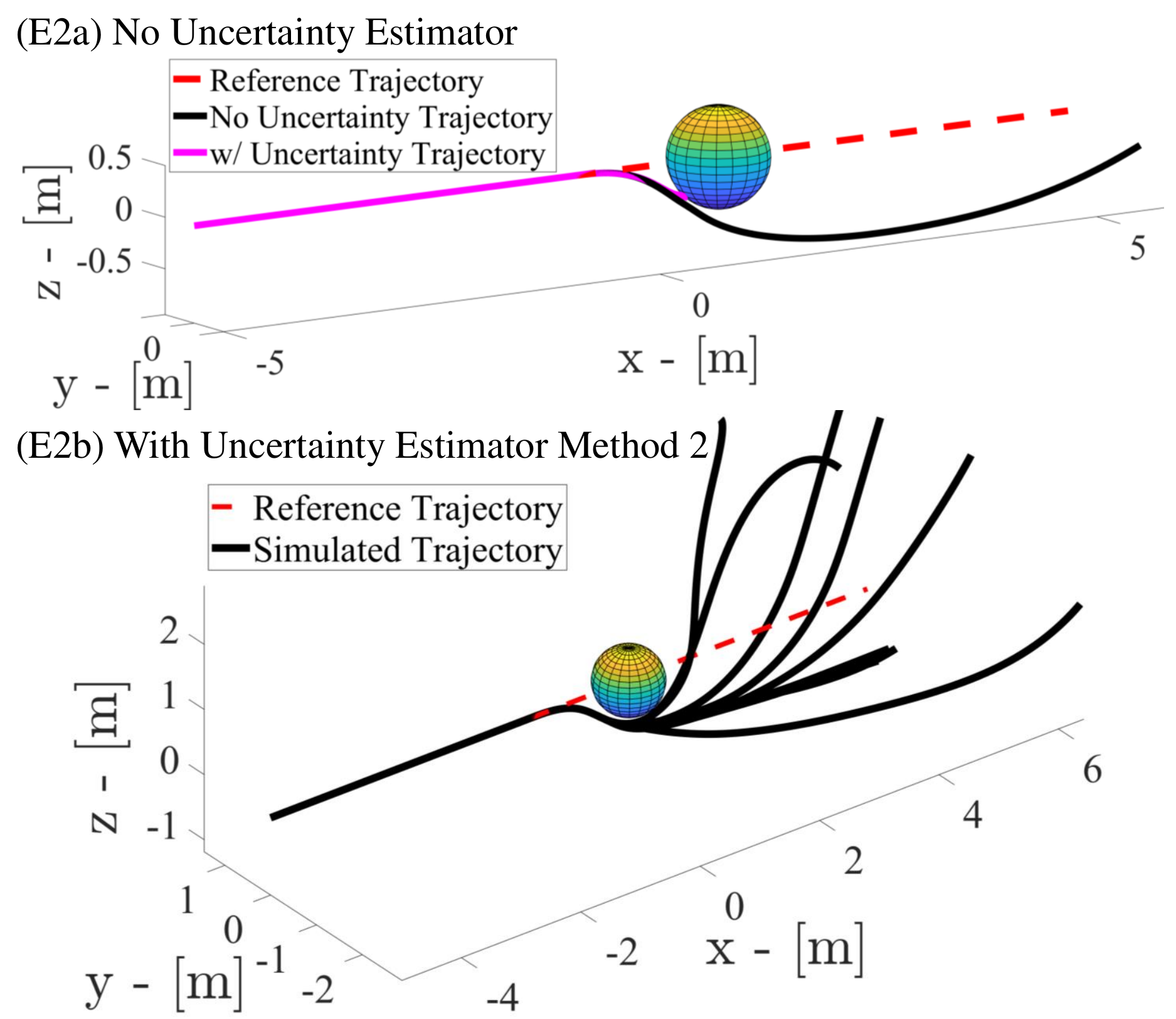}
	\caption{Simulated multirotor obstacle avoidance scenarios. Despite the nominal CBF being able to avoid the obstacle, the actual system with uncertainty \eqref{eq:uncertainty_drone} failed as shown in figure (E2a). Using the proposed uncertainty estimator (method 2), the augmented CBF can avoid the obstacle in 10 variations of uncertainties}
	\label{fig:result_drone}
\end{figure}
    
An optimization-based safety filter can be obtained by solving \eqref{CBF-QP}. We introduce state and input-dependent uncertainty in the system dynamics to model physical phenomena like aerodynamic interactions, input delays, compliance, etc. Therefore, error dynamics \eqref{eq:drone_dynamics} is more accurately
\begin{equation}
    \dot{\eta} = \underbrace{\mathcal{F}\eta}_{\hat{f}(\eta)} + \underbrace{\mathcal{G}B_u u}_{\hat{g}(\eta)u} + \underbrace{\Delta_{A}(\eta) + \Delta_{B}(\eta)u}_{\Delta(\eta,u)}. \label{eq:uncertainty_drone}
\end{equation}
In particular, we simulated the presence of the following state-dependent disturbance:
\begin{align} 
    \Delta_{A}(\eta) =
    [\,0_{1\times 6},\, c_d\tanh{(\dot{p}^T)},\, 0_{1\times 4}
   ]^T,
\end{align}
where $c_d$ can be interpreted as drag coefficients. 
Moreover, we also inject input uncertainty as 
${ \Delta_{B}(\eta) = \mbox{diag}([\delta_{u_1}, \delta_{u_2}, \delta_{u_3},  \delta_{u_4}]),
}$
where ${\delta_{u_i} \in (-1,0]}$ is the percent reduction of input representing delays. Since higher-order barrier functions are adopted for this scenario, we employ Method 2 to address model uncertainty. Incorporating the proposed disturbance observer \eqref{eq:bhat}-\eqref{eq:xidot} estimator into the CBF constraint, we have
${ \dot{h}_{e,i} = \dot{\hat{h}}_{e,i} + \frac{\partial h_{e,i}}{\partial \eta}(\hat{\Delta} + e)}$, where ${\dot{\hat{h}}_{e,i}}$ represents the time derivative of CBF along with the nominal system dynamics, and one can verify symbolically that ${\frac{\partial h_{e,i}}{\partial \eta} \neq 0}$. Therefore, we leverage Lemma~\ref{lem:Vb} and Corollary~\ref{theo:met3} to obtain a robust control barrier constraint with uncertainty estimation. 

For simulation, we choose parameters such that the uncertainty Lipschitz constant ${\delta_L = 0.1}$ and the global upper bound ${\delta_b = 0.1}$, and the results are depicted in Fig. \ref{fig:result_drone}. The nominal system with a CLF-QP controller filters with a barrier-based safety filter can avoid the obstacle (shown in (E2a)), and the disturbed system under a nominal approach is unable to main safety. We then incorporated the proposed disturbance observer \eqref{eq:bhat}-\eqref{eq:xidot} and augment the barrier constraint using \eqref{heineq}. The resulting controller can avoid the obstacle with different disturbances (different coefficients/parameters while satisfying ${\delta_L}$ and ${\delta_b}$) as presented in (E2b), showcasing the versatility of the proposed uncertainty estimator.
\vspace{-1mm}
\section{Conclusions and Future Work}
\vspace{-1mm}
We proposed a novel uncertainty estimator for control affine systems with unmodeled, state, and input-dependent uncertainty to improve system robustness. Moreover, the bounds for the estimation error and output were introduced. By incorporating the estimator with CBFs, robust safety conditions were derived. We showcase two methods to assert safety with theoretical analysis and simulation validations. 

Future work aims to remove the knowledge of the uncertainty bound $\delta_b$, and Lipschitz bound $\delta_L$ by using machine learning techniques like Gaussian processes \cite{dacs2023} to estimate these parameters from data. 
We also seek to combine state observer-based CBF methods \cite{agrawal2022safe} with our methods. 
Lastly, hardware validation and comparison with our prior work \cite{akella2022learning} will further showcase the performance of these methods.

\bibliographystyle{IEEEtran}
\bibliography{References.bib}

\end{document}